\documentclass[a4paper,12pt]{article}
\usepackage{amsmath}
\usepackage{amsthm}
\usepackage{amsfonts}
\usepackage{mathrsfs}
\usepackage{graphicx}
\usepackage{hyperref}

\newtheorem{theorem}{Theorem}
\newtheorem{lemma}{Lemma}

\numberwithin{equation}{section}

\hypersetup{ pdftitle={Draft},
pdfauthor={Fiki T. Akbar, Bobby E. Gunara}, 
pdfkeywords={supersymmetry, local existence}, 
bookmarksnumbered, pdfstartview={FitH}, urlcolor=blue, }

\begin{document}
\pagestyle{plain}




\title{\LARGE\textbf{Bosonic Part of $4d$  $N=1$ Supersymmetric Gauge Theory with General Couplings: Local Existence}}

\author{{Fiki T. Akbar$^{\sharp}$, Bobby E. Gunara}$^{\flat,\sharp}$, Triyanta$^{\flat,\sharp}$, Freddy P. Zen$^{\flat,\sharp}$ \\ \\
$^{\flat}$\textit{\small Indonesian Center for Theoretical and
Mathematical Physics (ICTMP)}
\\ {\small and} \\
$^{\sharp}$\textit{\small Theoretical Physics Laboratory}\\
\textit{\small Theoretical High Energy Physics and Instrumentation Research Group,}\\
\textit{\small Faculty of Mathematics and Natural Sciences,}\\
\textit{\small Institut Teknologi Bandung}\\
\textit{\small Jl. Ganesha no. 10 Bandung, Indonesia, 40132}\\
\small email: ft\_akbar@students.itb.ac.id, bobby@fi.itb.ac.id, triyanta@fi.itb.ac.id, fpzen@fi.itb.ac.id}

\date{}

\maketitle




\begin{abstract}

In this paper, we prove the local existence of the bosonic part of
$N=1$ supersymmetric gauge theory in four dimensions with general
couplings. We start with the Lagrangian of the vector and chiral multiplets 
with general couplings and scalar potential turned on.
Then, for the sake of simplicity, we set all fermions vanish at
the level of equations of motions, so we only have the bosonic
parts of the theory. We apply Segal's general theory to show the
local existence of solutions of equations of motions by
taking K\"ahler potential to be bounded above by $U(n)$ symmetric
K\"ahler potential and the first derivative of gauge couplings to
be at most linear growth functions.

\end{abstract}




\section{Introduction}

In 1963, I. Segal  has developed a method to prove the existence
and uniqueness of the solutions of a semi-linear evolution equations
using a semi-group \cite{Segal1}. Then, ten years later, he
applied it to prove the existence of local and global solutions
for four dimensional Yang-Mills equations in the temporal gauge
condition \cite{Segal2}. Such a study has been extended to the
case of Yang-Mills theory coupled to scalar fields in three
dimensions \cite{Ginibre}, and in four dimensions \cite{Eardley,
Eardley2}. \\
\indent Our interest here is to extend the results in
\cite{Eardley} to the case of minimal ($N=1$) supersymmetric
Yang-Mills theory coupled chiral multiplets. This theory has
become a prominent subject over four decades since it might
provide solutions to major problems in the Standard Model of
particle physics such as the unification of gauge couplings and
the hierarchy problems.\\
\indent In this paper, we prove the local existence of solutions
of $N = 1$ supersymmetry gauge theory in four dimensions with
general couplings. Our starting point is to consider $N=1$
Lagrangian consists of chiral and vector with general couplings
such that we have a nonlinear $\sigma$-model with K\"ahler
metric, general analytic gauge kinetic functions
determined by holomorphic functions, and the scalar potential.
Then, we derive field equations of motions and setting all the
fermionic field to be zero at this level for the sake of
simplicity. Thus, we only have an effective bosonic theory that
describes the interaction between the bosonic field $(\phi,A)$
where $\phi$ is the complex scalar fields and $A$ is the gauge
fields. \\
\indent By assuming that the temporal component of gauge fields
vanish (analogue with the temporal gauge in Yang-Mills theory) and
introducing new fields $(\pi,E)$, one can then transform the
equation of motions into a semi-linear form which contains linear
and non-linear terms. For our analysis, we take the fields
$u=(A,E,\phi,\pi ,\bar{\phi},\bar{\pi})$ lying in
$\mathcal{H}=(H_{2} \times H_{1})^{3}$ where $H_{p}$ denotes a
Sobolev space. We show that the linear terms are globally defined
in $\mathcal{H}$ and generates a one parameter semigroup. Finally,
by Segal general theory \cite{Segal1}, the local existence of
semi-linear evolution equation is established by showing that the
non-linear parts satisfy local Lipshitz condition.\\
\indent Since we have a generalized semi-linear evolution
equation, we have to take several assumptions on the general
couplings such that Segal's general theory can be used to our
problem. First, we assume that K\"ahler
potential is bounded above by $U(n)$ symmetric K\"ahler potential
and we derive several estimates for K\"ahler potential and
Christoffel symbol. These estimates can be used to eliminate the
quantity associated with K\"ahler metric in our analysis.
Second, we take some conditions on the gauge kinetic
couplings, namely the derivative of the gauge coupling must be at
most a linear growth.  Finally, we assume that the scalar
potential has to be at least $C^{3}$-functions and its derivative
is locally Lipshitz function.\\
\indent Another problem that arises is the constraint equation which
can be solved by the technique developed in \cite{Eardley}. This
technology can be mentioned in order. We firstly decompose the $E$
field into unique transverse ($E_{T}$) and longitudinal parts
($E_{L}$). Then by introducing a new field $E_{C}$ such that
$E_{C}=E_{L}$ if the constraints fulfilled, we modify the original
equation of motions by replacing $E_{L}$ with $E_{C}$. Using the
above mentioned conditions on general couplings, we prove that the
non-linear parts is locally Lipshitz function. Hence, they admit
local solutions. At the end, we show that the solutions of the
modified equations with the constraints satisfied are the solution
of the original equation of motions.\\
\indent The organization of the paper can be mentioned as follow.
We shortly review a four dimensional $N=1$ supersymmetric gauge
theory in which the vector multiplets are coupled to
arbitrary chiral multiplets in Section 2. Section 3 is devoted to
discuss several aspects of field equation of motions including a
modification of equation of motions to solve the constraint
problems. In section 4, we discuss the internal scalar
manifold and derive several estimates. In section 5, we prove that
the non-linear part of equation of motions satisfies Lipshitz
condition and finally prove the local existence.

\section{General Couplings of Chiral and Vector Multiplets}
\label{GCCVMs}

In this section, we review shortly  four dimensional $N=1$
supersymmetric gauge theory in which the vector multiplets are
coupled generally to arbitrary chiral multiplets. Here, we only write
terms which are useful for our analysis in the paper. For an
excellent review, interested reader can further consult,
for example, \cite{Wess,DF,ADF}.\\
\indent The theory consists of $n_v$ vector multiplets,
$(A^a_{\mu}, \lambda^a)$ coupled to $n_c$ chiral multiplets,
$(\phi^i, \chi^i)$ where the Latin alphabets $a,b =1,...,n_v$, and
$i,j=1,...,n_c$ show the number of multiplets, while the Greek
alphabets $\mu,\nu =0,...,3$ show the spacetime indices. In the
vector multiplets we have gauge fields $A^a_{\mu}$ together with
their fermionic partners $\lambda^a$. On the other side, the
chiral multiplets contain complex scalars $\phi^i$ and their
fermionic partners $\chi^i$.\\
\indent Furthermore, $N=1$ supersymmetry demands the following
conditions: First, the scalars $\phi^i$ span a K\"ahler manifold
endowed with metric $g_{i\bar{j}} = \partial_i
\partial_{\bar{j}}K $ where $K \equiv K(\phi, \bar{\phi})$
is a real function called K\"ahler potential. Second, there
exists a set of holomorphic functions, namely $(f_{ab}, X^i_a, W)$
which  are gauge couplings, Killing vectors, and a superpotential,
respectively. Finally, the existence of a real function called
scalar potential which can be written as
\begin{equation}
V  =  g^{i\bar{j}} \, \partial_i W
\partial_{\bar{j}} \bar{W} +
\frac{1}{8}h^{ab} P_{a} P_{b} \ , \label{potensial}
\end{equation}
where $P_a$ are real functions called Killing potentials (or
momentum maps), determined by $X^i_a$ via
\begin{equation}
X^{i}_{a} = \frac{i}{2}g^{i\bar{j}} \, \partial_{\bar{j}}P_{a} \ ,
\label{Killingeq}
\end{equation}
with $h^{ab}$ is the inverse of $h_{ab} \equiv
{\mathrm{Re}}f_{ab}$. Then, one can write down the bosonic part of
the $N=1$ Lagrangian as
\begin{equation}
\mathcal{L}  =
-g_{i\bar{j}}D^{\mu}\phi^{i}D_{\mu}\bar{\phi}^{\bar{j}} -
\frac{1}{4}h_{ab}\mathcal{F}^{a}_{\mu\nu}\mathcal{F}^{b\mu\nu}+\frac{1}{4}k_{ab}\mathcal{F}^{a}_{\mu\nu}\tilde{\mathcal{F}}^{b\mu\nu}
- V \ , \label{Lagrange}
\end{equation}
where $k_{ab} \equiv {\mathrm{Im}}f_{ab}$, the covariant
derivative $D_{\mu}\phi^{i} \equiv
\partial_{\mu}\phi^{i} + X^{i}_{a}A^{a}_{\mu}$, and the gauge field strength $\mathcal{F}^{a}_{\mu\nu}
\equiv \partial_{\mu}A^{a}_{\nu}-\partial_{\nu}A^{a}_{\mu} +
f^{a}_{bc}A^{b}_{\mu}A^{c}_{\nu}$. The dual field
$\tilde{\mathcal{F}}^{a\mu\nu}$ is defined as
$\tilde{\mathcal{F}}^{a\mu\nu}  \equiv
\frac{1}{2}\epsilon^{\mu\nu\rho\sigma}\mathcal{F}_{\rho\sigma}^{a}$.
It is worth to mention that Lagrangian (\ref{Lagrange}) is
invariant under the following supersymmetry transformation of the
fields up to three-fermion terms
\begin{eqnarray}
\delta\lambda_{\bullet}^a &=& \frac{1}{2}
\left({\mathcal{F}}^a_{\mu\nu}
- {\mathrm{i}}\widetilde{\mathcal{F}}^a_{\mu\nu}\right) \gamma^{\mu\nu}\epsilon_{\bullet}
+ 2 {\mathrm{i}} h^{ab} P_b \, \epsilon_{\bullet} \;, \nonumber\\
\delta\chi^i &=& {\mathrm{i}}\partial_{\nu} \phi^i \, \gamma^{\nu}
\epsilon^{\bullet} + 2 g^{i\bar{j}} \partial_{\bar{j}} \bar{W} \epsilon_{\bullet} \;,\label{susytr}\\
\delta A^a_{\mu} &=& \frac{\mathrm{i}}{2}
\bar{\lambda}_{\bullet}^a \gamma_{\mu}\epsilon^{\bullet} +
\frac{\mathrm{i}}{2} \bar{\epsilon}_{\bullet}
\gamma_{\mu} \lambda^{\bullet a}\;, \nonumber\\
\delta \phi^i &=& \bar{\chi}^i \epsilon_{\bullet} \;.\nonumber
\end{eqnarray}
Additionally, in the theory one can replace $P_a$ by $P_a + \xi_a$
where $\xi_a$ are real constants which give rise Fayet-Iliopoulos
term.

In our analysis, we assume that the scalar potential to be at least a $C^{2}$ function and a satisfies the local Lipshitz condition
\begin{equation}
\|\partial_{j}V(\phi')-\partial_{j}V(\phi)\| \leq C(\|\phi'\|,\|\phi\|)\|\phi'-\phi\| \:, \label{ScalarPotential}
\end{equation}
where $C(\|\phi'\|,\|\phi\|)$ is a bounded function depend on
$\|\phi\|$. The condition above implies that the holomorphic
superpotential has to be at least a $C^{3}$ function.

\section{Field Equations of Motions}
\label{FEMs}

This section is devoted to discuss several aspect of field
equation of motions. In particular, we take all fermions to be
trivial at this level. Thus, the gauge fields and the scalars are
the main ingredients of our analysis in this paper. To simplify
the analysis we take a condition $A_0 = 0$ which is not a gauge
condition since in general Lagrangian (\ref{Lagrange}) is no
longer gauge invariant \footnote{For renormalizable
Yang-Mills-Higgs theory with constant $f_{ab}$ and $g_{i\bar{j}}$
the condition $A_0 = 0$ is called temporal gauge, see for example,
\cite{Akbar}.}.\\
\indent First of all, the gauge field equation of motions is given
by
\begin{equation}
D_{\mu}\left(h_{ab}\mathcal{F}^{b\mu\nu}-k_{ab}\tilde{\mathcal{F}}^{b\mu\nu}\right)
= g_{i\bar{j}}\left(X^{i}_{a}D^{\nu}\bar{\phi}^{\bar{j}} +
\bar{X}^{\bar{j}}_{a}D^{\nu}\phi^{i}\right) \ , \label{gaugeEOM}
\end{equation}
with
\begin{equation}
D_{\mu}\left(h_{ab}\mathcal{F}^{b\mu\nu}\right) =
\partial_{\mu}\left(h_{ab}\mathcal{F}^{b\mu\nu}\right) +
f_{ac}^{d}A^{c}_{\mu}\left(h_{db}\mathcal{F}^{b\mu\nu}\right) \ .
\end{equation}
The field strength tensor $\mathcal{F}^{a}_{\mu\nu}$ satisfies
Bianchi identity,
\begin{equation}
D_{\mu}\tilde{\mathcal{F}}^{b\mu\nu}=0 \ . \label{Bianchiid}
\end{equation}
By defining
\begin{equation}
E^{as}=-\mathcal{F}^{a0s}=-\partial^{0}A^{ar} \ ,
\end{equation}
and using (\ref{Bianchiid}), we can rewrite (\ref{gaugeEOM}) as
\begin{eqnarray}
\frac{\partial E^{ar}}{\partial t} & = & \partial_{s}\partial^{s}A^{ar} - \partial^{r}\partial_{s}A^{as} + f^{a}_{bc}\partial_{s}\left(A^{as}A^{ar}\right) + h^{ab}\mathcal{F}^{csr}\partial_{s}h_{bc} + h^{ab}h_{de}f^{d}_{bc}A^{c}_{s}\mathcal{F}^{esr}\nonumber\\
& & -h^{ab}E^{cr}\partial_{0}h_{bc}-h^{ab}\tilde{\mathcal{F}}^{c0r}\partial_{0}k_{bc} - h^{ab}\tilde{\mathcal{F}}^{csr}\partial_{s}k_{bc} - h^{ab}\left(k_{de}f^{d}_{bc} - k_{bd}f^{d}_{ce}\right)A^{c}_{s}\mathcal{F}^{esr}\nonumber\\
& & -h^{ab}g_{i\bar{j}}\left(X^{i}_{a}D^{r}\bar{\phi}^{\bar{j}} + \bar{X}^{\bar{j}}_{a}D^{r}\phi^{i}\right),
\end{eqnarray}
together with the constraint equations
\begin{equation}
\mathcal{C}^a(t) = -\partial_{s}E^{as} + 4\pi \rho^a \ ,
\label{constraint}
\end{equation}
where
\begin{eqnarray}
4\pi\rho^{a} & = & h^{ab}g_{i\bar{j}}\left(X^{i}_{b}D^{0}\bar{\phi}^{\bar{j}} + \bar{X}^{\bar{j}}_{b}D^{0}\phi^{i}\right) + h^{ab}\tilde{\mathcal{F}}^{cs0}\partial_{s}k_{bc} - h^{ab}E^{cs}\partial_{s}h_{bc} \nonumber\\
& & h^{ab}\left(k_{de}f^{d}_{bc} -
k_{bd}f^{d}_{ce}\right)A^{c}_{s}\tilde{\mathcal{F}}^{es0} -
h^{ab}h_{de}f^{d}_{bc}A^{c}_{s}E^{es} \ ,\label{rho}
\end{eqnarray}
with the initial value of (\ref{constraint}) is
$\mathcal{C}^a(0)=0$.
Next, we consider the scalar field equation of motions which have
been modified into
\begin{eqnarray}
\frac{\partial \pi^{i}}{\partial t} & = & \partial_{s}\partial^{s}\phi^{i} + \Gamma^{i}_{kl}\left(\partial_{r}\phi^{k}\partial^{r}\phi^{l} - \pi^{k}\pi^{l}\right) + X^{i}_{a}\partial_{r}A^{ar} + A^{ar}\partial_{r}\phi^{k}\nabla_{k}X^{i}_{a} \nonumber\\
& & + g^{i\bar{j}}g_{k\bar{l}}A^{ar}D_{r}\phi^{k}\nabla_{\bar{j}}\bar{X}^{\bar{l}}_{a} -g^{i\bar{j}}\left(\frac{1}{4}h_{ab\bar{j}}\mathcal{F}^{a}_{\mu\nu}\mathcal{F}^{b\mu\nu} -
\frac{1}{4}k_{ab\bar{j}}\mathcal{F}^{a}_{\mu\nu}\tilde{\mathcal{F}}^{b\mu\nu} +
\partial_{\bar{j}}V_{S}\right) \ , \label{scalarEOM}
\end{eqnarray}
where
\begin{eqnarray}
\pi^{i} & = & -D^{0}\phi^{i} \ , \\
\nabla_{k}X^{i}_{a} & = & \partial_{k}X^{i}_{a} +
\Gamma^{i}_{kl}X^{l}_{a} \ ,
\end{eqnarray}
together with its complex conjugate.\\
\indent Now, we rewrite the equations of motion (\ref{gaugeEOM})
and (\ref{scalarEOM}) into the following form
\begin{equation}
\frac{du}{dt} = \mathcal{A}u + J(u) \ ,\label{originaleq}
\end{equation}
where
\begin{equation}
u = \left[\begin{array}{c}
A^{as}\\
E^{as}\\
\phi^{i}\\
\pi^{i}\\
\bar{\phi}^{\bar{j}}\\
\bar{\pi}^{\bar{j}}\end{array}\right],\quad\quad\quad\quad \mathcal{A}u =\left[\begin{array}{c}
E^{as}\\
\partial_{r}\partial^{r}A^{as}-\partial_{r}\partial^{s}A^{ar}\\
\pi^{i}\\
\partial_{r}\partial^{r}\phi^{i}\\
\bar{\pi}^{\bar{j}}\\
\partial_{r}\partial^{r}\bar{\phi}^{\bar{j}}\end{array}\right] \ ,\nonumber
\end{equation}
and
\begin{equation}
J(u) = \left[\begin{array}{c}
 0\\
f_{bc}^{a}\partial_{r}\left(A^{bs}A^{cr}\right) + h^{ab}\left(\mathcal{F}^{esr}\partial_{r}h_{bc} - \tilde{\mathcal{F}}^{esr}\partial_{r}k_{bc}+h_{de}f^{d}_{bc}A^{c}_{r}\mathcal{F}^{esr}\right)+\mathcal{D}_{1}\\ 
0\\ 
\Gamma^{i}_{kl}\left(\partial_{s}\phi^{k}\partial^{s}\phi^{l}-\pi^{k}\pi^{l}\right)
+X_{a}^{i}\partial_{s}A^{as} +
A^{as}\nabla_{k}X^{i}_{a}\partial_{s}\phi^{k} + g^{i\bar{j}}g_{l\bar{k}}A^{a}_{s}\nabla_{\bar{j}}\bar{X}^{\bar{k}}_{a}D^{s}\phi^{l}+\mathcal{D}_{2}\\ 
0\\ 
\Gamma^{\bar{j}}_{\bar{k}\bar{l}}\partial_{s}\bar{\phi}^{\bar{k}}\partial^{s}\bar{\phi}^{\bar{l}}
+\bar{X}_{a}^{\bar{j}}\partial_{s}A^{as} +
A^{as}\nabla_{\bar{k}}\bar{X}^{\bar{j}}_{a}\partial_{s}\bar{\phi}^{\bar{k}}
-g^{i\bar{j}}g_{k\bar{l}}A^{a}_{s}\nabla_{i}X^{k}_{a}\left\{\partial^{s}\bar{\phi}^{\bar{l}}
+A^{bs}\bar{X}^{\bar{l}}_{b}\right\}+\bar{\mathcal{D}}_{2}\end{array}\right] \ . 
\end{equation}
where,
\begin{eqnarray}
\mathcal{D}_{1} & = &
-h^{ab}\left(h_{bci}\pi^{i}+h_{bc\bar{j}}\bar{\pi}^{\bar{j}}\right)E^{cs}
-h^{ab}\left(k_{bci}\pi^{i}+k_{bc\bar{j}}\bar{\pi}^{\bar{j}}\right)\tilde{\mathcal{F}}^{c0s}\nonumber\\
& &  - h^{ab}\left(k_{de}f^{d}_{bc}-k_{bd}f^{d}_{ce}\right)A^{c}_{r}\tilde{\mathcal{F}}^{esr}
- h^{ab}g_{i\bar{j}}\left(X^{i}_{b}D^{s}\bar{\phi}^{\bar{j}}+X^{\bar{j}}_{b}D^{s}\phi^{i}\right) \ ,\\
\mathcal{D}_{2} & = &
g^{i\bar{j}}\frac{h_{ab\bar{j}}}{4}\left(-2E^{a}_{s}E^{bs} - \mathcal{F}^{a}_{rs}\mathcal{F}^{brs}\right)
+ g^{i\bar{j}}\frac{k_{ab\bar{j}}}{2}\epsilon^{srl}E^{a}_{s}\mathcal{F}^{b}_{rl}-g^{i\bar{j}}\partial_{\bar{j}}V \ ,\\
\bar{\mathcal{D}}_{2} & = &
g^{i\bar{j}}\frac{h_{abi}}{4}\left(-2E^{a}_{s}E^{bs} -
\mathcal{F}^{a}_{rs}\mathcal{F}^{brs}\right) +
g^{i\bar{j}}\frac{k_{abi}}{2}\epsilon^{srl}E^{a}_{s}\mathcal{F}^{b}_{rl}-g^{i\bar{j}}\partial_{i}V
\ ,
\end{eqnarray}
together with the constraint equation (\ref{constraint}).\\
\indent To solve the constraint problem, we use the method in
\cite{Eardley} which can be structured as follows. First, we
modify (\ref{gaugeEOM}) by decomposing $E$ field into unique
transverse parts (divergence free) $E_{T}$ and longitudinal parts
(curl free) $E_{L}$
\begin{equation}
E = E_{T} + E_{L} \ , \label{decompostitonEField}
\end{equation}
with
\begin{equation}
\partial_{s}E^{s}_{T} = 0,\qquad \varepsilon_{qrs}\partial^{r}E^{s}_{L} = 0 \
.
\end{equation}
Then, replacing $E_{L}$ with the new fields, $E_{C}$ which equals to $E_{L}$ when the constraint is
satisfied (see lemma \ref{LemmaConstraint}),
\begin{equation}
E^{s}_{L}\rightarrow E^{s}_{C} =
\partial_{s}\left(-\frac{1}{4\pi R}\ast\rho\right) \ ,
\end{equation}
where  $-\frac{1}{4\pi r}\ast\rho$ represents convolution of $\rho$
with the fundamental solution of Poisson equation,
\begin{equation}
-\frac{1}{4\pi R}\ast\rho = -\frac{1}{4\pi} \int_{\mathbb{R}^{3}}\:dx'\left(\frac{\rho(x')}{|x-x'|}\right).
\end{equation}

Let $\mathcal{H}=(H_{2}\times H_{1})^{3}$, where $H_{p}$ represent a Sobolev space of square integrable functions over $\mathbb{R}^{3}$ with their derivative up to order $p$ are also square integrable and let $\|\:.\:\|_{H_{p}}$ represents a Sobolev norm defined as
\begin{equation}
\|u\|_{H_{p}} =\left[\sum _{|\alpha |=0}^{p} \|D^{\alpha}u\|_{L_{2} }^{2} \right]^{\frac{1}{2}}.
\end{equation}
where $\alpha=(\alpha_{1},\alpha_{2},\ldots,\alpha_{m})$ is a multi-index of non-negative integers and $|\alpha| = \alpha_{1} + \alpha_{2} + \ldots + \alpha_{m}$. The $H_{p}$ space is also a Hilbert space. In our analysis, we take the fields $u = (A,E,\phi,\pi ,\bar{\phi},\bar{\pi}) $ to lie in $\mathcal{H}$.

Now we prove the following lemmas which are an extension of lemmas in \cite{Eardley} to incorporate the gauge kinetic functions and show that for fields $u$ lying in ${\mathcal{H}}$, the condition $E_{L}=E_{C}$ implies that $E$ satisfies the constraint equation.
\begin{lemma}
\label{LemmaSource}
Let $\rho$ is defined as in \ref{rho}. If the gauge kinetic function is also lies at least in $H_{1}$, then $E_{C}^{s} = \partial^{s}\left(-\frac{1}{4\pi r}\ast\rho\right) = \frac{1}{4\pi}\left(\frac{\hat{r}}{r^{2}}\ast\rho\right)^{s}$.
\end{lemma}

\begin{proof}

Since $\frac{1}{r}$ is a weak $L_{3,w}(\mathbb{R}^{3})$ function, then using generalize Young inequality for convolution product \cite{Reed}, we have
\begin{equation}
\left\|\frac{1}{r} \ast \rho \right\|_{L_{p}}\leq C\|\rho\|_{L_{q}}\left\|\frac{1}{r}\right\|_{3,w}\leq C'\|\rho\|_{L_{q}},
\end{equation}
for $\frac{1}{q}=\frac{1}{p}+\frac{2}{3}$. It means $\frac{1}{r}\ast \rho \in L_{p}(\mathbb{R}^{3})$ provided that $\rho \in L_{q}(\mathbb{R}^{3})$. Since all field $(A,E,\phi,\pi ,\bar{\phi},\bar{\pi})$ at least lie in $H_{1}$ and using the assumption that the gauge kinetic function is also lie at least in $H_{1}$, then by Sobolev embedding theorem, they all lie in $L_{s}$ for $2\leq s \leq 6$. Then, $\rho \in L_{q}$ for $1\leq q \leq 3$, thus $\frac{1}{r}\ast \rho \in L_{p}(\mathbb{R}^{3})$ for $p>3$.

Let $V$ be an arbitrary $C^{\infty}$ rapidly decreasing vector field in $\mathbb{R}^{3}$, we have
\begin{eqnarray}
\left\langle V_{s},E_{C}^{s} \right\rangle & = & \frac{1}{4\pi} \left\langle V_{s},-\partial^{s}\left(\frac{1}{r}\ast\rho\right) \right\rangle \nonumber\\
& = & \frac{1}{4\pi} \left\langle \partial^{s}V_{s},\left(\frac{1}{r}\ast\rho\right) \right\rangle \nonumber\\
& = & \frac{1}{4\pi} \left\langle \frac{1}{r}\ast\partial^{s}V_{s},\rho\right\rangle.
\end{eqnarray}
Since $\left|\frac{\hat{r}}{r}\right|$ is a weak $L_{3/2,w}(\mathbb{R}^{3})$ function and $\rho \in L_{q}(\mathbb{R}^{3})$ for $1\leq q \leq 3$, then
\begin{equation}
\left|\frac{\hat{r}}{r}\right|\ast\left|\rho\right| \in L_{r},
\end{equation}
for all $\frac{3}{2}<r<3$. Hence, we have,
\begin{eqnarray}
\frac{1}{4\pi} \left\langle \frac{1}{r}\ast\partial^{s}V_{s},\rho\right\rangle & = & \frac{1}{4\pi}\int\:V_{s}\left[\int \frac{x^{s}-x'^{s}}{\left|x-x'\right|^{3}}\rho(x')dx'\right]dx \nonumber\\
& = & \frac{1}{4\pi}\left\langle V_{s},\left(\frac{\hat{r}}{r^{2}}\ast\rho\right)^{s}\right\rangle.
\end{eqnarray}
And finally we have $E_{C}^{s} = \partial^{s}\left(-\frac{1}{4\pi r}\ast\rho\right) = \frac{1}{4\pi}\left(\frac{\hat{r}}{r^{2}}\ast\rho\right)^{s}$ which complete the proof.

\end{proof}

\begin{lemma}
\label{LemmaEc}
$E_{C}\in L_{p}$ for all $\frac{3}{2}<p<\infty$. In particular $E_{C}\in L_{2}$.
\end{lemma}

\begin{proof}

From the proof of the previous lemma, we already have $\rho \in L_{q}(\mathbb{R}^{3})$ for $1\leq q \leq 3$. Then using the representation of $E_{C}$ in previous lemma,
\begin{equation}
\|E_{C}\|_{L_{p}} = \left\|\frac{\hat{r}}{r^{2}}\ast\rho\right\|_{L_{p}} \leq C \|\rho\|_{L_{q}}\left\|\frac{\hat{r}}{r^{2}}\right\|_{3/2,w} \leq C'\|\rho\|_{L_{q}},
\end{equation}
for $1<p,q<\infty$ and $\frac{1}{q}=\frac{1}{p}+\frac{1}{3}$, then $E_{C} \in L_{s}$ for all $\frac{3}{2}<s<\infty$.

In particular, $E_{C} \in L_{2}$ and the following inequality holds,
\begin{equation}
\|E_{C}\|_{L_{2}}\leq C\|\rho\|_{L_{6/5}}. \label{lema1}
\end{equation}

\end{proof}

\begin{lemma}
\label{LemmaConstraint}
The condition $E_{C}=E_{L}$ is equivalent to $\partial_{s}E^{s}=4\pi\rho$.
\end{lemma}

\begin{proof}

Let $V$ be an arbitrary $C^{\infty}$ rapidly decreasing function in $\mathbb{R}^{3}$, we have
\begin{eqnarray}
\left\langle V,\partial_{s}E_{C}^{s} \right\rangle & = & \frac{1}{4\pi}\int\:V(x)\partial_{s}\left[\int \frac{x^{s}-x'^{s}}{\left|x-x'\right|^{3}}\rho(x')dx'\right]dx\nonumber\\
& = & -\frac{1}{4\pi} \int\:\rho(x')\left[\int \frac{x^{s}-x'^{s}}{\left|x-x'\right|^{3}}\partial_{s}V\: dx\right]dx'\nonumber\\
& = & \frac{1}{4\pi} \int\:\rho(x')\left[\int V(x)\partial_{s}\left(\frac{x^{s}-x'^{s}}{\left|x-x'\right|^{3}}\right)dx\right]dx'\nonumber\\
& = & \int \rho(x') V(x') dx'\nonumber\\
& = & \left\langle V,\rho \right\rangle,
\end{eqnarray}
where we used integrating by parts and exchange the order of integration. Thus, $E_{C}$ satisfies $\partial_{s}E_{C}^{s}=4\pi\rho$
as distribution.

Since $E_{C} \in L^{2}$, we can define it's Fourier transform and decomposing as
\begin{equation}
\hat{E}_{C} = \hat{E}_{C}^{T} + \hat{E}_{C}^{L},
\end{equation}
with
\begin{eqnarray}
\left(\hat{E}_{C}^{T}\right)^{s} & = & \left(\delta^{sr}-\frac{k^{s}k^{r}}{|k|^{2}}\right)\left(\hat{E}_{C}\right)_{r}\\
\left(\hat{E}_{C}^{L}\right)^{s} & = & \frac{k^{s}k^{r}}{|k|^{2}}\left(\hat{E}_{C}\right)_{r}.
\end{eqnarray}
Furthermore, $E_{C}$ is a gradient, then $E_{C}$ has vanishing curl, then $\hat{k}\times\hat{E_{C}}=0$ which implies that $E_{C}$ has only a longitudinal component, thus
\begin{equation}
\left(E_{C}^{L}\right)^{s} =  \frac{k^{s}k^{r}}{|k|^{2}}\left(\hat{E}_{C}\right)_{r} = \left(\hat{E}_{C}\right)^{s},
\end{equation}
and because $E_{C}$ satisfies $\partial_{s}E_{C}^{s}=4\pi\rho$, taking a Fourier transform, we have
\begin{equation}
\hat{E}_{C}^{s}=-i\frac{k^{s}}{|k|^{2}}\left(4\pi\hat{\rho}\right).
\end{equation}

Now suppose that the field $u$ satisfies the constraint equation,
\begin{equation}
\partial_{s}E^{s}=\partial_{s}E_{L}^{s}=4\pi\rho,
\end{equation}
where we used the decomposition in (\ref{decompostitonEField}). Taking a Fourier transform of the constraint equation, we get
\begin{eqnarray}
\hat{E}_{L}^{s} & = & -i\frac{k^{s}}{|k|^{2}}\left(4\pi\hat{\rho}\right)\nonumber\\
& = & \hat{E}_{C}^{s}.
\end{eqnarray}
Thus any solution of the constraint has $E_{L} = E_{C}$. Conversely, if $E_{L} = E_{C}$, the constraint is satisfied.

\end{proof}

With the modification, we can rewrite the equation of motions as
follows
\begin{equation}
\frac{du}{dt} = \mathcal{A}u + J(u),\label{eom}
\end{equation}
where
\begin{equation}
u = \left[\begin{array}{c}
A^{as}\\
E^{as}\\
\phi^{i}\\
\pi^{i}\\
\bar{\phi}^{\bar{j}}\\
\bar{\pi}^{\bar{j}}\end{array}\right],\quad\quad\quad\quad \mathcal{A}u =\left[\begin{array}{c}
E_{T}^{as}\\
\partial_{r}\partial^{r}A^{as}-\partial_{r}\partial^{s}A^{ar}\\
\pi^{i}\\
\partial_{r}\partial^{r}\phi^{i}\\
\bar{\pi}^{\bar{j}}\\
\partial_{r}\partial^{r}\bar{\phi}^{\bar{j}}\end{array}\right],\nonumber
\end{equation}
and
\begin{equation}
J(u) = \left[\begin{array}{c}
 \partial^{s}\{-\frac{1}{4\pi r}\ast\rho\}\\
f_{bc}^{a}\partial_{r}\left(A^{bs}A^{cr}\right) + h^{ab}\left(\mathcal{F}^{esr}\partial_{r}h_{bc} - \tilde{\mathcal{F}}^{esr}\partial_{r}k_{bc}+h_{de}f^{d}_{bc}A^{c}_{r}\mathcal{F}^{esr}\right)+\mathcal{D}_{1}\\ 
0\\ 
\Gamma^{i}_{kl}\left(\partial_{s}\phi^{k}\partial^{s}\phi^{l}-\pi^{k}\pi^{l}\right)
+X_{a}^{i}\partial_{s}A^{as} +
A^{as}\nabla_{k}X^{i}_{a}\partial_{s}\phi^{k} + g^{i\bar{j}}g_{l\bar{k}}A^{a}_{s}\nabla_{\bar{j}}\bar{X}^{\bar{k}}_{a}D^{s}\phi^{l}+\mathcal{D}_{2}\\ 
0\\ 
\Gamma^{\bar{j}}_{\bar{k}\bar{l}}\partial_{s}\bar{\phi}^{\bar{k}}\partial^{s}\bar{\phi}^{\bar{l}}
+\bar{X}_{a}^{\bar{j}}\partial_{s}A^{as} +
A^{as}\nabla_{\bar{k}}\bar{X}^{\bar{j}}_{a}\partial_{s}\bar{\phi}^{\bar{k}}
-g^{i\bar{j}}g_{k\bar{l}}A^{a}_{s}\nabla_{i}X^{k}_{a}\left\{\partial^{s}\bar{\phi}^{\bar{l}}
+A^{bs}\bar{X}^{\bar{l}}_{b}\right\}+\bar{\mathcal{D}}_{2}\end{array}\right]. 
\end{equation}

\section{Scalar Internal Manifold}
\label{SIM}

This section is assigned for the discussion of the internal scalar
manifold. In particular, we consider the case of the K\"ahler
potential to be bounded to a function and derive an estimates for
K\"ahler potential and Christoffel symbol. Our estimates derived
in this section is important in our analysis for proving the local
existence of (\ref{eom}).

As mention in section \ref{GCCVMs}, in four dimensions, the $N=1$ supersymmetry theory demands that the scalar field $(\phi,\bar{\phi})$ span a K\"ahler manifold with K\"ahler potential $K \equiv K(\phi, \bar{\phi})$. In this paper, we consider the case where the K\"ahler potential bounded above by $U(n_{c})$ symmetric K\"ahler potential and satisfies several conditions,
\begin{eqnarray}
K \leq \Phi(|\phi|) \: ,\label{KahlerPotential}\\
\left|\Gamma\right| \leq |\tilde{\Gamma}| \: ,\label{ChristoffelSymbol}
\end{eqnarray}
where $|\phi| =
\left(\delta_{i\bar{j}}\phi^{i}\bar{\phi}^{\bar{j}}\right)^{\frac{1}{2}}$ and $\tilde{\Gamma}$ is the Christoffel symbol of $\tilde{g}$.

We prove a lemma about estimates of K\"ahler potential and Christoffel symbol,
\begin{lemma}
\label{LemmaSIM}
Let $\mathcal{M}$ be a K\"ahler manifold with Kahler potential $K=K(\phi,\bar{\phi})$. If $\mathcal{M}$ satisfies (\ref{KahlerPotential}), (\ref{ChristoffelSymbol}) and
\begin{equation}
\left|\frac{F'}{2|\phi|}\right|  \leq  \epsilon \: ,\label{ConditionKahler}
\end{equation}
where $F(|\phi|) =  \frac{1}{4|\phi|^{2}}\left(\Phi''-\frac{\Phi'}{|\phi|}\right)$ with $\Phi' = \partial \Phi/\partial |\phi|$ and $\epsilon$ is a non negative constant, then we have the following estimates
\begin{eqnarray}
\left|K\right| & \leq  & \frac{\epsilon}{6} \left|\phi\right|^{6} +
\frac{C_{1}}{2}\left|\phi\right|^{4} + C_{2}\left|\phi\right|^{2} +
C_{3} \: ,\\
\left|\Gamma\right| & \leq & 2\epsilon |\phi|^{3} + C_{1}|\phi| \: .
\end{eqnarray}
\end{lemma}

\begin{proof}

Let $\tilde{\mathcal{M}}$ be a K\"ahler manifold generated by $\Phi$. We can write the metric $\tilde{g}_{i\bar{j}} = \partial_{i}\partial_{\bar{j}}\Phi$ as
\begin{equation}
\tilde{g}_{i\bar{j}} = \frac{\Phi'}{2|\phi|}\delta_{i\bar{j}} +
\frac{1}{4|\phi|^{2}}\left(\Phi''-\frac{\Phi'}{|\phi|}\right)\delta_{k\bar{j}}\delta_{i\bar{k}}
\phi^{k}\bar{\phi}^{\bar{k}} \: ,
\end{equation}
where $\Phi' = \partial \Phi/\partial |\phi|$. The inverse of the metric can written as,
\begin{equation}
\tilde{g}^{i\bar{j}} = \frac{2|\phi|}{\Phi'}\delta^{i\bar{j}} -
\frac{2}{\Phi'|\phi|}\left(\frac{\Phi''-\frac{\Phi'}{|\phi|}}{\Phi''+\frac{\Phi'}{|\phi|}}\right)
\phi^{i}\bar{\phi}^{\bar{j}} \: .
\end{equation}

The norm of the Christoffell symbol is
\begin{eqnarray}
|\tilde{\Gamma}| & = &
\left(\tilde{g}^{j\bar{j}}\tilde{g}^{k\bar{k}}\tilde{g}_{i\bar{i}}\tilde{\Gamma}^{i}_{jk}
\bar{\tilde{\Gamma}}^{\bar{i}}_{\bar{j}\bar{k}}\right)^{\frac{1}{2}}\nonumber\\
& = & \left(\tilde{g}^{j\bar{j}}\tilde{g}^{k\bar{k}}\tilde{g}_{i\bar{i}}
\tilde{g}^{i\bar{l}}\partial_{j}\tilde{g}_{k\bar{l}} \: \tilde{g}^{l\bar{i}}\:\partial_{\bar{j}}\tilde{g}_{l\bar{k}}\right)^{\frac{1}{2}}\nonumber\\
& = & \left(\tilde{g}^{l\bar{l}}\tilde{g}^{j\bar{j}}\tilde{g}^{k\bar{k}}\partial_{j}\tilde{g}_{k\bar{l}}\:\partial_{\bar{j}}\tilde{g}_{l\bar{k}}\right)^{\frac{1}{2}}\nonumber\\
& = & |\partial \tilde{g}| \: ,
\end{eqnarray}
and the first derivative of the metric is
\begin{equation}
\partial_{j}\tilde{g}_{k\bar{l}} =
F
\left(\delta_{k\bar{l}}\delta_{j\bar{i}}+\delta_{k\bar{i}}\delta_{j\bar{l}}\right)\bar{\phi}^{\bar{i}}
+ \frac{F'}{2|\phi|}
\delta_{k\bar{i}}\delta_{i\bar{l}}\delta_{j\bar{m}}\phi^{i}\bar{\phi}^{\bar{i}}\bar{\phi}^{\bar{m}} \: ,
\end{equation}
where
\begin{eqnarray}
F(|\phi|)& = & \frac{1}{4|\phi|^{2}}\left(\Phi''-\frac{\Phi'}{|\phi|}\right) \: ,\nonumber\\
\frac{F'}{2|\phi|} & = &
\frac{1}{8|\phi|^{3}}\left(\Phi'''-\frac{3\Phi''}{|\phi|}+\frac{3\Phi'}{|\phi|^{2}}\right) \: .
\end{eqnarray}

If condition (\ref{ConditionKahler}) is satisfied, then using the following inequality for integral,
\begin{equation}
\left|\int f(x) dx \right| \leq \int \left|f(x)\right| dx \: ,
\end{equation}
we have the following estimates,
\begin{eqnarray}
\left|F\right| & \leq  & \epsilon \left|\phi\right|^{2}+C_{1},\nonumber\\
\left|\Phi\right| & \leq  & \frac{\epsilon}{6} \left|\phi\right|^{6} +
\frac{C_{1}}{2}\left|\phi\right|^{4} + C_{2}\left|\phi\right|^{2} +
C_{3} \: ,
\end{eqnarray}
where $C_{1} = |F(0)|$, $C_{2} =
\left|\frac{\Phi'}{2|\phi|}(0)\right|$ and $C_{3} = |\Phi(0)|$.
Then we have the norm of the Christoffel symbol satisfies
\begin{equation}
\left|\tilde{\Gamma}\right| \leq 2\epsilon |\phi|^{3} + C_{1}|\phi| \: .
\end{equation}
Hence, by our assumption in \ref{KahlerPotential} and \ref{ChristoffelSymbol}, we have
\begin{eqnarray}
\left|\Phi\right| & \leq  & \frac{\epsilon}{6} \left|\phi\right|^{6} +
\frac{C_{1}}{2}\left|\phi\right|^{4} + C_{2}\left|\phi\right|^{2} +
C_{3} \: ,\\
\left|\Gamma\right| & \leq & 2\epsilon |\phi|^{3} + C_{1}|\phi| \: .
\end{eqnarray}
This complete the proof.
\end{proof}

Our assumption in (\ref{ConditionKahler}) is satisfied for several examples of K\"ahler manifold, for examples are $\mathbb{C}^{n}$ and $\mathbb{C}P^{n}$ which are widely used in the theory. For $\mathbb{C}^{n}$, the K\"ahler potential is given by $|\phi|^{2}$, then clearly $F$ is vanish, hence $\frac{F'}{2|\phi|}$ is bounded by 0. In case of $\mathbb{C}P^{n}$, the K\"ahler potential (using standard Fubini-Study metric) is given by
\begin{equation}
\Phi_{\mathbb{C}P^{n}}(|\phi|) = \ln(1+|\phi|^{2}) \: .
\end{equation}
Then we have,
\begin{equation}
F = -\frac{1}{\left(1+|\phi|^{2}\right)^{2}} \: ,
\end{equation}
and
\begin{equation}
\left|\frac{F'}{2|\phi|}\right| = \frac{2}{\left(1+|\phi|^{2}\right)^{3}} \: ,
\end{equation}
which is bounded above by 2.

\section{Local Existence}

In this section, we will prove the local existence of the evolution
equation (\ref{eom}) using Segal's theorem. Furthermore, we shall
show that solutions of (\ref{eom})  are the solutions of the
original equations, namely (\ref{gaugeEOM}) and
(\ref{scalarEOM}).\\
\indent In section \ref{FEMs}, we have derived the equation
of motions for $u$ and have taken the field $u$ lie in
$\mathcal{H}=(H_{2}\times H_{1})^{3}$. Let us consider the
linear part of the evolution equation (\ref{eom}),
\begin{equation}
\frac{du}{dt} =\mathcal{A}u \: . \label{LinearizedEquations1}
\end{equation}
By decomposing the $A$ and $E$ fields into transverse and longitudinal components, we can write (\ref{LinearizedEquations1}) as follow,
\begin{eqnarray}
\frac{d}{dt} \left[\begin{array}{c}
A^{as}_{T}\\
E^{as}_{T}\\ \end{array}\right]  = \begin{bmatrix}
0 & I \\
\triangle & 0
\end{bmatrix} \left[\begin{array}{c}
A^{as}_{T}\\
E^{as}_{T}\\ \end{array}\right]& \: , &\qquad
\frac{d}{dt} \left[\begin{array}{c}
A^{as}_{L}\\
E^{as}_{L}\\ \end{array}\right]  =  0 \: ,\\ 
\frac{d}{dt} \left[\begin{array}{c}
\phi^{i}\\
\pi^{i}\\ \end{array}\right]  =  \begin{bmatrix}
0 & I \\
\triangle & 0
\end{bmatrix} \left[\begin{array}{c}
\phi^{i}\\
\pi^{i}\\ \end{array}\right]& \: , & \qquad
\frac{d}{dt} \left[\begin{array}{c}
\bar{\phi}^{\bar{j}}\\
\bar{\pi}^{\bar{j}}\\ \end{array}\right]  =  \begin{bmatrix}
0 & I \\
\triangle & 0
\end{bmatrix} \left[\begin{array}{c}
\bar{\phi}^{\bar{j}}\\
\bar{\pi}^{\bar{j}}\\ \end{array}\right] \: . \label{LinearizedEquations2}
\end{eqnarray}
Each pair of fields $(A_{T},E_{T})$, $(\phi,\pi)$,and $(\bar{\phi},\bar{\pi})$ satisfies the linear wave equation
and $(A_{L},E_{L})$ is a constant of the linearized equation. Thus, the linear operator $\mathcal{A}$ generates
a one-parameter semigroup on $\mathcal{H}$ and for any initial value $u_{0}=u(t_{0}) \in \mathcal{H}$,
 the linearized equation admits a classical solution which can be written as \footnote{for details review, see \cite{Segal1}, \cite{Zheng}.},
\begin{equation}
u(t) = e^{\mathcal{A}(t-t_{0})}u_{0} \: .
\end{equation}
Then, it follows that the solution of linearized equation is globally defined on $\mathcal{H}$.

Following the result above, by writing the evolution equation (\ref{eom}) as an integral equation,
\begin{equation}
u(t)=e^{\mathcal{A}(t-t_{0})}u_{0} +\int _{t_{0} }^{t} \, ds\, e^{\mathcal{A}(s-t_{0})}J(u(s)) \: , \label{IntegralEq}
\end{equation}
the local existence of solution of the equation is established by showing that the nonlinear operator $J$ satisfies Lipshitz condition,
\begin{equation}
\|J(u')-J(u)\| \leq C\left(\|u'\|,\|u\|\right)\|u'-u\| \: ,
\end{equation}
for all $u',u \in \mathcal{H}$. The norm $\|\:.\:\|$ is designed for $\mathcal{H}$ norm and $C(,)$ is some monotonically increasing,
finite function of the norm indicated. Then, for any initial data $u_{0}=u(t_{0})\in D_{\mathcal{A}}$,
where $D_{\mathcal{A}}$ is a domain of linear operator $\mathcal{A}$,
the evolution equation (\ref{eom}) admits a unique classical solution on some
interval $(T_{1},T_{2})$ containing $t_{0}$ either $(T_{1},T_{2}) = (-\infty,\infty)$ or $\|u(t)\|\rightarrow \infty$ as $t\rightarrow T_{1}$ or $T_{2}$.

Let us write the components of the non-linear operator $J$ as $J=(J_{1} ,J_{2} ,J_{3} ,J_{4} ,J_{5} ,J_{6} )$.
The proof that $J$ satisfies a Lipshitz condition is facilitated by a Schauder ring property for Sobolev space
over ${\mathbb R}^{3} $ and the Sobolev inequality over $\mathbb{R}^{3}$,
\begin{equation}
\|\partial ^{r} u\|_{L_{p}} \leq C\|\partial^{s} u\|_{L_{m}}^{\theta} \, \|u\|_{L_{q}}^{1-\theta} \: ,
\end{equation}
for real numbers $q,m$ with $1\le q,m\le \infty $ and $r,s$ are integers where $0\le r<s$ which satisfy
\begin{equation}
\frac{1}{p} =\frac{r}{3} +\theta\left(\frac{1}{m} -\frac{s}{3} \right)+(1-\theta)\frac{1}{q} \: ,
\end{equation}
with $r/s\le \theta\le 1$ and $p$ is non negative, and a constant
$C$ depends only on $m,j,q,r$ and $\theta$. For further discussion
on Sobolev inequality, see \cite{Adams}.\\
\indent Now, since the theory has scalar fields dependent gauge
couplings, we have to make an assumption for the gauge kinetic
function in order to prove that $E_{C}$ is a mapping from
$\mathcal{H}$ to $H_{2}$ and locally Lipshitz continuous.

\begin{lemma}\label{LemmaEC}
Let the fields $(A,E,\phi,\pi ,\bar{\phi},\bar{\pi}) \in {\mathcal{H}}$. If the first derivative of the real part of the gauge kinetic function is at most a linear growth,
\begin{equation}
\|\partial_{i}h_{ab}\|_{H_{1}} = \|\partial_{\bar{j}}h_{ab}\|_{H_{1}} \leq C \|\phi\|_{H_{2}} \: , \label{GaugeKineticCondition}
\end{equation}
then $E_{C}$ lies in $H_{2}$. Furthermore, if $\partial_{s}h_{ab}$ is also a locally Lipshitz function, then $E_{C}$ is a locally Lipshitz function.
\end{lemma}

\begin{proof}
Define a norm for a gauge indexed field as
\begin{equation}
|A|^{2} = h_{ab}A^{as}A^{b}_{s} \: .
\end{equation}
By definition of Sobolev norm, we have
\begin{equation}
\|E_{C}\|_{H_{2}} = \left(\|E_{C}\|^{2}_{L_{2}}+\|D^{2}E_{C}\|^{2}_{L_{2}}\right)^{1/2} \: .
\end{equation}
Using (\ref{lema1}), we have
\begin{equation}
\|E_{C}\|_{L_{2}}\leq C\|\rho\|_{L_{6/5}} \: ,
\end{equation}
and by definition of $\rho$ and using Holder inequality, we have
\begin{eqnarray}
\|\rho\|^{2}_{L_{6/5}} & \leq & \|E A\|^{2}_{L_{6/5}} + \|X \pi\|^{2}_{L_{6/5}} + \|E \partial_{s}h\|^{2}_{L_{6/5}} + \|\mathcal{F} \partial_{s}k\|^{2}_{L_{6/5}} + \|kA\mathcal{F}\|^{2}_{L_{6/5}}\nonumber\\
& \leq & \|E\|^{2}_{L_{2}}\|A\|^{2}_{L_{3}} + \|X\|^{2}_{L_{3}}\|\pi\|^{2}_{L_{2}} + \|E\|^{2}_{L_{2}}\|\partial_{s}h\|^{2}_{L_{3}} + \|\partial_{s}k\|^{2}_{L_{3}}\|\mathcal{F}\|^{2}_{L_{2}} + \|kA\|^{2}_{L_{3}}\|\mathcal{F}\|^{2}_{L_{2}}\nonumber\\
& \leq & \|E\|^{2}_{L_{2}}\|A\|^{2}_{H_{1}} + \|X\|^{2}_{H_{1}}\|\pi\|^{2}_{L_{3}} + \|E\|^{2}_{L_{2}}\|\partial_{s}h\|^{2}_{H_{1}} + \|\mathcal{F}\|^{2}_{L_{2}}\left(\|\partial_{s}k\|^{2}_{H_{1}} + \|kA\|^{2}_{H_{1}}\right) \: ,\nonumber\\
\end{eqnarray}
where we used the Sobolev inequality to show that $\|u\|_{L_{3}} \leq C \|u\|_{H_{1}}$.

Since $\partial_{s}E_{C}^{s}= 4\pi\rho$, by taking Fourier transform on both sides, we have
\begin{equation}
\hat{E}_{C}^{s} = -i4\pi\frac{k^{s}}{|k|^{2}}\hat{\rho} \: ,
\end{equation}
hence, we have
\begin{eqnarray}
\|D^{2}E_{C}\|_{L_{2}}^{2} & = & \int_{\mathbb{R}^{3}} d^{3}x \left(\partial_{q}\partial_{r}E_{Cs}\right)^{2}\nonumber\\
& = & \int_{\mathbb{R}^{3}} d^{3}k \left|k^{2}\hat{E}_{Cs}\right|^{2}\nonumber\\
& = & (4\pi)^{2}\int_{\mathbb{R}^{3}} d^{3}k \left|k_{s}\hat{\rho}\right|^{2}\nonumber\\
& = & (4\pi)^{2}\int_{\mathbb{R}^{3}} d^{3}x \left(\partial_{s}\rho\right)^{2} = (4\pi)^{2}\|\partial_{s}\rho\|_{L_{2}}^{2} \: .
\end{eqnarray}

Using definition of Sobolev norm and Schauder ring property, we have
\begin{eqnarray}
\|\partial_{s}\rho\|_{L_{2}}^{2} & \leq & \|E A\|^{2}_{H_{1}} + \|X \pi\|^{2}_{H_{1}} + \|E \partial_{s}h\|^{2}_{H_{1}} + \|\mathcal{F} \partial_{s}k\|^{2}_{H_{1}} + \|kA\mathcal{F}\|^{2}_{H_{1}}\nonumber\\
& \leq & \|E\|^{2}_{H_{1}}\|A\|^{2}_{H_{1}} + \|X\|^{2}_{H_{1}}\|\pi\|^{2}_{H_{1}} + \|E\|^{2}_{H_{1}}\|\partial_{s}h\|^{2}_{H_{1}} + \|\mathcal{F}\|^{2}_{H_{1}}\left(\|\partial_{s}k\|^{2}_{H_{1}} + \|kA\|^{2}_{H_{1}}\right) \: .\nonumber\\
\end{eqnarray}

By definition of field strength, we get an estimate
\begin{equation}
\|\mathcal{F}\|_{H_{1}} \leq C\left(\|A\|_{H_{2}} + \|A\|^{2}_{H_{2}}\right) \: , \label{FieldStrength}
\end{equation}
and using condition (\ref{GaugeKineticCondition}),
\begin{eqnarray}
\|\partial_{s}h\|_{H_{1}} & \leq & \|\partial_{i}h\|_{H_{1}}\|\partial_{s}\phi\|_{H_{1}} \: ,\nonumber\\
& \leq & C\|\phi\|^{2}_{H_{2}}
\end{eqnarray}
and the fact that gauge kinetic function is a holomorphic function, we have
\begin{eqnarray}
\|E_{C}\|_{H_{2}} & \leq & C \left(\|A\|_{H_{2}}\|E\|_{H_{1}} + \|\phi\|_{H_{2}}\|\pi\|_{H_{1}}  + \|\phi\|^{2}_{H_{2}}\|E\|_{H_{1}} + \|A\|_{H_{2}}\|\phi\|^{2}_{H_{2}}\right.\nonumber\\
& &  \left.+ \|A\|^{2}_{H_{2}}\|\phi\|^{2}_{H_{2}} + \|A\|^{3}_{H_{2}}\|\phi\|_{H_{2}}  + \|A\|^{2}_{H_{2}}\|\phi\|^{2}_{H_{1}}\right) \: ,
\end{eqnarray}
which shows that $E_{C}$ lies in $H_{2}$ for all $u \in \mathcal{H}$ and the first part of the lemma is proven.

Now we will prove the Lipshitz condition for $E_{C}$ which is important for proving the local existence of the evolution equation. From the definition of $E_{C}$ we have
\begin{eqnarray}
E_{C}^{as}(u')-E_{C}^{as}(u) & = & \frac{\hat{r}^{s}}{4\pi r^{2}}\ast\left(-f^{c}_{bf}\left\{h'^{ab}h'_{cd}A'^{f}_{r}E'^{r}_{d} - h^{ab}h_{cd}A^{f}_{r}E^{r}_{d}\right\}\right. \nonumber\\
& &\quad- h'^{ab}g'_{i\bar{j}}\left\{X'^{i}_{b}\bar{\pi}^{'\bar{j}} + \bar{X}_{'b}^{\bar{j}}\pi'^{i}\right\} +
h^{ab}g_{i\bar{j}}\left\{X^{i}_{b}\bar{\pi}^{\bar{j}} + \bar{X}_{b}^{\bar{j}}\pi^{i}\right\}\nonumber\\
& & \quad+ h'^{ab}\left(k'_{de}f^{d}_{bc} - k'_{bd}f^{d}_{ce}\right)A'^{cr}\tilde{\mathcal{F}}'^{es0} - h^{ab}\left(k_{de}f^{d}_{bc} - k_{bd}f^{d}_{ce}\right)A^{cr}\tilde{\mathcal{F}}^{es0}\nonumber\\
& &\quad \left. - h'^{ab}E'^{cs}\partial_{s}h'_{bc} + h^{ab}E^{cs}\partial_{s}h_{bc} + h'^{ab}\tilde{\mathcal{F}}'^{cs0}\partial_{s}k'_{bc} - h^{ab}\tilde{\mathcal{F}}^{cs0}\partial_{s}k_{bc}\right) \: .\nonumber\\
\end{eqnarray}
Recalling the estimate from the prove of the first part of lemma and using assumption that $\partial_{s}h_{ab}$ is locally Lipshitz, then we have
\begin{eqnarray}
\left\|E_{C}^{as}(u')-E_{C}^{as}(u)\right\|_{H_{2}} & \leq & K\left\{\|A'-A\|_{H_{2}}\left[\|E'\|_{H_{1}} + \|\phi\|_{H_{2}}^{2}\left( \|A\|_{H_{2}} + \|A\|_{H_{2}}\|A'+A\|_{H_{2}} \right.\right.\right.\nonumber\\
& &\left.\left.+\|A\|^{2}_{H_{2}}+\|A'+A\|_{H_{2}}\right)\right] + \|E'-E\|_{H_{1}}\left(\|A\|_{H_{2}} + \|\phi\|^{2}_{H_{2}}\right)\nonumber\\
& &+ \|\pi'-\pi\|_{H_{1}}\|\phi\|_{H_{2}} + \|\phi'-\phi\|_{H_{2}}\left[\|E\|_{H_{1}}\|\phi'+\phi\|_{H_{2}}\right. \nonumber\\
& &\left.\left.+ \|A\|_{H_{2}}\left(1+\|A\|_{H_{2}}\right)\|\phi'+\phi\|_{H_{2}} + \|A'\|^{2}_{H_{2}}\left(1+\|A'\|_{H_{2}}\right)\right]\right\}\nonumber\\
& = & C_{1}(\|u'\|,\|u\|)\|A'-A\|_{H_{2}} + C_{2}(\|u'\|,\|u\|)\|E'-E\|_{H_{1}} \nonumber\\
& &+ C_{3}(\|u'\|,\|u\|)\|\phi'-\phi\|_{H_{2}} + C_{4}(\|u'\|,\|u\|)\|\pi'-\pi\|_{H_{1}} \: .
\end{eqnarray}
Then for $u \in \mathcal{H}$, $E_{C}$ is a locally Lipshitz function which proves the second part of the lemma.

\end{proof}

The final proof of Lipshitz condition for $J$ is established by showing that the other components of $J$ are also Lipshitz functions. Using Sobolev inequality, we have
\begin{eqnarray}
\|J_{2}(u')-J_{2}(u)\|_{H_{1}} & \leq & C\bigg\{ \|A'-A\|_{H_{2}}\Big[ \|A\|_{H_{2}} + \|A\|^{2}_{H_{2}} + \|E\|_{H_{1}} \|\phi'\|^{2}_{H_{2}} +  \|\phi\|^{2}_{H_{2}} + \|E\|_{H_{1}}\|\phi'\|^{2}_{H_{2}} \nonumber\\
& &+\|A'+A\|_{H_{2}}\left(1+ \|A'\|_{H_{2}} + \|\phi\|^{2}_{H_{2}} + \|\phi'\|_{H_{2}}\|\pi'\|_{H_{1}}\right) + \|\phi'\|_{H_{2}}\|\pi'\|_{H_{1}} \Big]\nonumber\\
& & +\|E'-E\|_{H_{1}}\Big[ \|\phi'\|_{H_{2}} \|\pi'\|_{H_{1}} +  \|\phi\|^{2}_{H_{2}} +  \|A'\|_{H_{2}} \|\phi'\|^{2}_{H_{2}}\Big]\nonumber\\
& & +  \|\phi'-\phi\|_{H_{2}}\Big[ \|\phi'\|_{H_{2}}  + \|\phi'+\phi\|_{H_{2}}\left( \|E\|_{H_{1}} +  \|A'\|_{H_{2}} + \|A'\|^{2}_{H_{2}} +  \|A\|_{H_{2}} \right)\nonumber\\
& &+\|\pi\|_{H_{1}}\|E\|_{H_{1}} + \|\pi\|_{H_{2}}\left(\|A\|_{H_{2}} + \|A\|^{2}_{H_{2}}\right) + \|\phi\|_{H_{2}}\Big]\nonumber\\
& &+\|\pi'-\pi\|_{H_{1}}\Big[\|E\|_{H_{1}} + \|\phi'\|_{H_{2}}\left(\|A\|_{H_{2}} + \|A\|^{2}_{H_{2}}\right)\Big]\bigg\}\nonumber\\
& = & C_{5}(\|u'\|,\|u\|)\|A'-A\|_{H_{2}} + C_{6}(\|u'\|,\|u\|)\|E'-E\|_{H_{1}} \nonumber\\
& &\quad+ C_{7}(\|u'\|,\|u\|)\|\phi'-\phi\|_{H_{2}} + C_{8}(\|u'\|,\|u\|)\|\pi'-\pi\|_{H_{1}} \: ,\label{J2}
\end{eqnarray}
and
\begin{eqnarray}
\left\|J_{4}(u')-J_{4}(u)\right\| _{H_{1}} & \leq & C\left\{\left\|\Gamma'-\Gamma\right\| _{H_{2} } \Big[\|\phi'\|_{H_{2}} + \|A'\|_{H_{2}} \|\phi'\|^{2}_{H_{2}} + \|A'\|^{2}_{H_{2}}\|\phi'\|^{2}_{H_{2}}+\|\pi\|^{2}_{H_{1}}\Big]\right. \nonumber\\
& & \|E'-E\|_{H_{1}}\Big[\|\phi\|_{H_{2}}\|E'+E\|_{H_{1}} + \|\phi'\|_{H_{2}}\|\mathcal{F}\|_{H_{1}}\Big] + \|\pi'-\pi\|_{H_{1}}\|\Gamma\|_{H_{2}} \nonumber\\
& &\|A'-A\|_{H_{2}}\Big[\|\phi'\|_{H_{2}} + \|A'-A\|_{H_{2}}\left(\|\phi'\|_{H_{2}} + \|\Gamma\|_{H_{2}}\|\phi\|^{2}_{H_{2}}\right) + \|\Gamma\|_{H_{2}}\|\phi\|^{2}_{H_{2}}\nonumber\\
& & +\left(1 + \|A'+A\|_{H_{2}}\right)\left(\|\phi\|_{H_{2}}\|\mathcal{F}'-\mathcal{F}\|_{H_{1}} + \|\phi\|_{H_{2}}\|E'\|_{H_{1}}\right)\Big]\nonumber\\
& & + \|\phi'-\phi\|_{H_{2}}\Big[\|\Gamma\|_{H_{2}}\|\phi'+\phi\|_{H_{2}}\left(1+\|A\|^{2}_{H_{2}}\right) + \|A'\|_{H_{2}} + \|A\|_{H_{2}} \nonumber\\
& & \left.+ \|A\|^{2}_{H_{2}} + \|E\|_{H_{2}}\|\mathcal{F}\|_{H_{1}} + \|E\|^{2}_{H_{2}} + \|\mathcal{F}\|^{2}_{H_{2}}\Big] + \|\partial_{\bar{j}}V'_{S}-\partial_{\bar{j}}V_{S}\|_{H_{1}} \right\} \: .\nonumber\\
\end{eqnarray}
Using estimate (\ref{ChristoffelSymbol}) and (\ref{FieldStrength}) and using the assumption in (\ref{ScalarPotential}), we have
\begin{eqnarray}
\left\|J_{4}(u')-J_{4}(u)\right\| _{H_{1}} & \leq & C_{9}(\|u'\|,\|u\|)\|A'-A\|_{H_{2}} + C_{10}(\|u'\|,\|u\|)\|E'-E\|_{H_{1}} \nonumber\\
& &+ C_{11}(\|u'\|,\|u\|)\|\phi'-\phi\|_{H_{2}} + C_{12}(\|u'\|,\|u\|)\|\pi'-\pi\|_{H_{1}} \: . \label{J4}
\end{eqnarray}
Thus, from lemma \ref{LemmaEC}, (\ref{J2}) and (\ref{J4}), the nonlinear operator $J$ is a mapping from $\mathcal{H}$ to itself and satisfies locally Lipshitz condition.

Then by Segal's theorem, for any initial data $u_{0}$ in $\mathcal{H}$, there exists a positive constant $T>0$ depending on $u_{0}$ such that the equation (\ref{eom}) admits a unique mild solution which is continuous in $\mathcal{H}$ for an interval $[0,T]$, i.e. $u \in C([0,T],\mathcal{H})$ which satisfies (\ref{IntegralEq}). Furthermore, the solutions can be extended into maximal mild solutions on interval $[0,T_{\mathrm{max}})$ such that either
\begin{enumerate}
\item $T_{\mathrm{max}} = +\infty$ and the equation (\ref{eom}) admits a global solution, or
\item $\|u(t)\|\rightarrow \infty$ as $t\rightarrow T_{\mathrm{max}}$ and the solution blow up on a finite time $T_{\mathrm{max}}$ \: .
\end{enumerate}

If the initial value $u_{0}$ lies in $D_{\mathcal{A}}$ then equation (\ref{eom}) admits a unique classical solutions $u(t)$ for an interval $[0,T_{\mathrm{max}})$ which remains in $D_{\mathcal{A}}$ and satisfies differential equations
\begin{equation}
\frac{du(t)}{dt} = \mathcal{A}u(t) + J(u(t)) \: ,
\end{equation}
with $\frac{du(t)}{dt}$ is a continuous curve in $\mathcal{H}$. Then, the solutions $u$ belong to,
\begin{equation}
u \in C^{1}\left([0,T_{\mathrm{max}}),\mathcal{H}\right) \cap C\left([0,T_{\mathrm{max}}),D_{\mathcal{A}}\right) \: ,
\end{equation}
such that either $T_{\mathrm{max}} = +\infty$ or $\|u(t)\|\rightarrow \infty$ as $t\rightarrow T_{\mathrm{max}}$.

Now, we shall show that the solution of modified equation (\ref{eom}) which satisfy the constraint is the solution of original equation.

Following result of \cite{Eardley} for the constraint equation, we have
\begin{equation}
\|\mathcal{C}(t)\|^{2}_{L_{2}} \leq  \|\mathcal{C}(0)\|^{2}_{L_{2}}\exp\left(\int_{0}^{t}C'\|E(s)\|_{H_{1}}\right) \: .
\end{equation}
Since the initial value of the constraint equation is $\mathcal{C}(0)=0$, then $\mathcal{C}(t)$ is vanished in the interval existence of $u\left(t\right)$. Hence, the solutions of modified equation (\ref{eom}) always satisfy the constraint equation then it is the solutions of the original equation.

Therefore, we have proven,
\begin{theorem}
Let $u_{0}$ be any initial data lying in $\mathcal{H}=(H_{2}\times H_{1})^{3}$. If the conditions
(\ref{ScalarPotential}), (\ref{KahlerPotential}), (\ref{ChristoffelSymbol}), and (\ref{ConditionKahler}) are satisfied,
 then there exists a positive constant $T_{\mathrm{max}}>0$ depending on $u_{0}$ such that the integral
 equation (\ref{IntegralEq}) admits a unique maximal solution $u(t)$ on interval $[0,T_{\mathrm{max}})$
 which belongs to $u \in C\left([0,T_{\mathrm{max}}),\mathcal{H}\right)$ and either
\begin{enumerate}
\item $T_{\mathrm{max}} = +\infty$ and the equation (\ref{eom}) admits a global solution, or
\item $\|u(t)\|\rightarrow \infty$ as $t\rightarrow T_{\mathrm{max}}$ and the solution blow up on a finite time $T_{\mathrm{max}}$ \: .
\end{enumerate}

Furthermore, if $u_{0}$ lies in $D_{\mathcal{A}}$ and satisfies
the constraint $\mathcal{C}(u_{0})=0$, then the equation
(\ref{eom}) admits a unique classical solution $u(t)$ for an
interval $[0,T_{\mathrm{max}})$ which remains in
$D_{\mathcal{A}}$, and belong to $u \in
C^{1}\left([0,T_{\mathrm{max}}),\mathcal{H}\right) \cap
C\left([0,T_{\mathrm{max}}),D_{\mathcal{A}}\right)$, and satisfy
the constraint $\mathcal{C}(u(t))=0$ such that either
$T_{\mathrm{max}} = +\infty$ or $\|u(t)\|\rightarrow \infty$ as
$t\rightarrow T_{\mathrm{max}}$.
\end{theorem}

\appendix

\section{Convention and Notation}
The purpose of this appendix is to inform our conventions used in this paper. The spacetime
metric is flat with the signature $(-,+,+,+)$.

The following indices are used:
\begin{description}
\item[] $\mu,\nu,\rho,\sigma = 0,\ldots,3$, \hspace{30mm} label
4-dimensional flat spacetime \item[] $r,s,p,q = 1,2,3$,
\hspace{35mm} label 3-dimensional flat space \item[]
$i,\bar{i},j,\bar{j},k,\bar{k} = 1,\ldots,n_{c}$, \hspace{22.5mm}
label $n_{c}$ dimensional K\"ahler manifold \item[] $a,b,c,d =
1,\ldots,n_{v}$, \hspace{29mm} label the gauge index
\end{description}

\section{List of Inequality}
In this appendix, we mention some basic inequalities used in this paper. For detail reviews, see \cite{Reed, Adams}.

\subsection*{Young inequality for convolution}

Suppose $f \in L_{p}(\mathbb{R}^{d})$ and $g \in L_{q}(\mathbb{R}^{d})$ and
\begin{equation}
\frac{1}{p} + \frac{1}{q} = \frac{1}{r} + 1 \: ,
\end{equation}
with  $p,q,r \in \mathbb{R}$  and $1\leq p,q,r \leq \infty$. Then
\begin{equation}
\|f\ast g\|_{L_{r}} \leq \|f\|_{L_{p}}\|g\|_{L_{q}} \: .
\end{equation}

If $g \in L_{q,w}(\mathbb{R}^{d})$ where $L_{q,w}$ is weak $L_{q}$ space, then
\begin{equation}
\|f\ast g\|_{L_{r}} \leq \|f\|_{L_{p}}\|g\|_{L_{q,w}} \: ,
\end{equation}
with $p,q,r$ are defined and satisfy condition above.

\subsection*{Holder inequality}

Let $p,q,r$ be positive numbers and satisfy $p,q,r \leq 1$ and
\begin{equation}
\frac{1}{p} + \frac{1}{q} = \frac{1}{r}.
\end{equation}
Suppose $f \in L_{p}(\mathbb{R}^{d})$ and $g \in L_{q}(\mathbb{R}^{d})$, then $fg \in L_{r}(\mathbb{R}^{d})$ and
\begin{equation}
\|fg\|_{L_{r}} \leq \|f\|_{L_{p}}\|g\|_{L_{q}} \: .
\end{equation}

\subsection*{Sobolev inequality}

Let $q,m$ be real numbers with $1\le q,m\le \infty $ and $r,s$ are integers where $0\le r<s$ which satisfy
\begin{equation}
\frac{1}{p} =\frac{r}{3} +\theta\left(\frac{1}{m} -\frac{s}{3} \right)+(1-\theta)\frac{1}{q} \: ,
\end{equation}
with $r/s\le \theta\le 1$ and $p$ is non negative. For $u \in H_{s}(\mathbb{R}^{3})\cap L_{q}(\mathbb{R}^{3})$, there is a positive constant $C$ which depends only on $m,j,q,r$ and $\theta$ such that the following inequality holds
\begin{equation}
\|\partial ^{r} u\|_{L_{p}} \leq C\|\partial^{s} u\|_{L_{m}}^{\theta} \, \|u\|_{L_{q}}^{1-\theta} \: .
\end{equation}

\section*{Acknowledgments}

Our research is supported by Hibah Kompetensi DIKTI 2012 No. 781a/I1.C01/PL/2012 and Riset Desentralisasi DIKTI-ITB 2012 No.003.8/TL-J/DIPA/SPK/2012.

\end{document}